\documentclass{amsart}
\usepackage[utf8]{inputenc}
\usepackage{lineno,hyperref}

\usepackage[english]{babel}
\usepackage{amsthm}
\usepackage{mathtools}
\usepackage{bbm} 
\usepackage{amsmath}
\usepackage{pbox}
\usepackage{geometry}
\usepackage{amsfonts}
\usepackage{graphicx}
\usepackage{amssymb}
\usepackage{enumitem}
\usepackage{booktabs}
\usepackage{tikz}
\usepackage{tikz-cd}
\usepackage{subcaption}
 \newcommand{\norm}[1]{\left\Vert #1 \right\Vert}
 \DeclareMathOperator{\diag}{diag}
 \newcommand{\linspan}[1]{\operatorname{span}\{#1\}}
\newtheorem{Theorem}{Theorem}[section]
\newtheorem{Corollary}[Theorem]{Corollary}
\newtheorem{Lemma}[Theorem]{Lemma}
\newtheorem{Proposition}[Theorem]{Proposition}
\theoremstyle{definition}
\newtheorem{Definition}[Theorem]{Definition}
\theoremstyle{remark}
\newtheorem{Remark}[Theorem]{Remark}
\newtheorem*{Example}{Example}
\numberwithin{equation}{section}

\title{On $\mathbb{Z}$-invariant self-adjoint extensions of the Laplacian on quantum circuits}

\author{A. Balmaseda}
\address{Depto. de Matemáticas, Univ. Carlos III de Madrid, Avda. de la Universidad 30, 28911 Leganés, Madrid, Spain}
\email{abalmase@math.uc3m.es}

\author{F. Di Cosmo}
\address{Depto. de Matemáticas, Univ. Carlos III de Madrid, Avda. de la Universidad 30, 28911 Leganés, Madrid, Spain, Instituto de Ciencias
Matemáticas (CSIC - UAM - UC3M - UCM) ICMAT}
\email{fcosmo@math.uc3m.es}

\author{J.M. Pérez-Pardo}
\address{Depto. de Matemáticas, Univ. Carlos III de Madrid, Avda. de la Universidad 30, 28911 Leganés, Madrid, Spain, Instituto de Ciencias
Matemáticas (CSIC - UAM - UC3M - UCM) ICMAT}
\email{jmppardo@math.uc3m.es}



\begin{document}

\begin{abstract}
    An analysis of the invariance properties of self-adjoint extensions of symmetric operators under the action of a group of symmetries is presented. For a given group $G$, criteria for the existence of $G$-invariant self-adjoint extensions of the Laplace-Beltrami operator over a Riemannian manifold are illustrated and critically revisited. These criteria are employed for characterising self-adjoint extensions of the Laplace-Beltrami operator on an infinite set of intervals, $\Omega$, constituting a quantum circuit, which are invariant under a given action of the group $\mathbb{Z}$. A study of the different unitary representations of the group $\mathbb{Z}$ on the space of square integrable functions on $\Omega$ is performed and the corresponding $\mathbb{Z}$-invariant self-adjoint extensions of the Laplace-Beltrami operator are introduced. The study and characterisation of the invariance properties allows for the determination of the spectrum and generalised eigenfunctions in particular examples.
\end{abstract}

\maketitle

\section{Introduction}
    The relation between Physics and symmetries has been successful and fruitful up to the point that Physical theories, from the most fundamental ones, like the Standard Model, to the effective ones applied, e.g., in condensed matter Physics, are intimately related with the symmetries and transformation properties of their underlying structures. For instance, gauge symmetries in the former case or crystallographic groups in the latter.

    One of the aims of this article is to provide a framework for a systematic analysis of the action of symmetry groups on the configuration space of quantum circuits. Even if the configuration space is invariant under the action of a given group, not all the possible self-adjoint extensions need to be compatible with that symmetry group. We will use the characterisation introduced in \cite{IbortLledoPerezPardo2015} to identify the set of self-adjoint extensions compatible with the action of the symmetry group in the particular case of infinite chains made up by repeating a finite block. This kind of periodic lattices are widely used in solid state physics as approximations for systems like crystals, when the period is much smaller than the size of the system \cite{Kittel2004,Bloch1929}. To have control on the symmetries that the system possess is also important in the determination of the spectrum and the spaces of eigenfunctions, as they will carry the same representation. The importance of this characterisation is that the space of mathematically possible self-adjoint extensions for a given quantum circuit is very large. As it will be clear in the discussion below and in the subsequent sections, the space of self-adjoint extensions contains all the possible topologies for a given graph and also many other situations that are not compatible with any given topology \cite{BalachandranBimonteMarmoSimoni1995,PerezPardoBarberoLinanIbort2015}.

   The study of symmetries in the context of quantum circuits is particularly relevant in relation with the development of new quantum information and processing devices. Superconducting qubits \cite{DevoretSchoelkopf2013,Wendin2017} are one of the promising technologies that can lead to scalable quantum computation. The framework of quantum circuits provides a natural setting to model and study them. In particular, a superconducting qubit can be seen as a truncation to the lowest energy levels of a corresponding quantum circuit \cite{WendinShumeiko2007a}. Quantum circuits are, in general, infinite dimensional quantum systems and this has inherent difficulties in their description. However, one of the sources of decoherence in the description of superconducting qubits arises precisely because of the aforementioned truncation to the lowest orders. Hence, it is worth to address the system in its full generality and try, for instance, to achieve general controllability results as was done in \cite{BalmasedaPerezPardo2019}.

   In the most abstract setting the dynamics of a quantum circuit can be described by the Schrödinger equation, where the Hamiltonian of the system will be given formally by a Laplacian operator defined over a disjoint union of intervals with, possibly, a scalar potential defined over the intervals. The topology of the circuit, which can be described mathematically by a directed graph, will arise from the boundary conditions that one implements at the boundaries of the intervals. These boundary conditions need to be fixed in order to determine a well-defined self-adjoint operator. Otherwise the dynamics of the system will not satisfy the unitarity of the evolution as characterised by the Stone-von Neumann Theorem.

    It was proven by G.~Grubb \cite{Grubb1968} that there is a one-to-one correspondence between self-adjoint extensions of the Laplace-Beltrami operator and boundary conditions defined in terms of pseudo-differential operators. We will take here the approach introduced in \cite{AsoreyIbortMarmo2005} and further developed in \cite{IbortLledoPerezPardo2015}, according to which the space of self-adjoint extensions of the Laplace operator is given by a unitary operator acting on the Hilbert space of boundary data. This approach is general enough to accommodate all the physically acceptable self-adjoint extensions. Moreover, the self-adjoint extensions characterised in this way are well suited for the numerical approximation of their spectrum, c.f. \cite{IbortPerezPardo2013,LopezYelaPerezPardo2017}. Another advantage is that the same ideas can be used to analyse other differential operators like the Dirac operator, cf. \cite{IbortPerezPardo2015,PerezPardo2017}, which are also important in the context of quantum circuits.

    The article is organised as follows. In Section \ref{sec:sa_extensions_and_symmetry} we introduce the definitions of invariance and symmetries that will be needed for the rest of the article. In this section we will also carry on with the analysis introduced in \cite{IbortLledoPerezPardo2015a} to cover other possible situations not covered by the theorems proved in that work. In Section \ref{sec:local-symmetries} we use the invariance properties to characterise a blockwise structure of the possible self-adjoint extensions that will represent a given quantum circuit. Finally in Section \ref{sec:global_symmetries} we study the case of an infinite quantum circuit, a quantum chain, that is made of the repetition of elemental finite blocks. Some discussion and outlook are given in Section \ref{sec:discussion}.


\section{Self-adjoint extensions with symmetries} \label{sec:sa_extensions_and_symmetry}
    In the Schrödinger-Dirac picture of Quantum Mechanics, one associates a Hilbert space $\mathcal{H}$ with a quantum system, and self-adjoint operators play the role of observables whereas pure states are elements of the projective Hilbert space $\mathcal{P}(\mathcal{H})$. Moreover, dynamics of ``closed'' quantum systems are described by strongly continuous one-parameter groups of unitary transformations, and according to Stone-von Neumann Theorem the generator of a strongly continuous one-parameter group of unitary transformations is a self-adjoint operator. Therefore, it is evident that self-adjointness plays a fundamental role in any quantum theory and a proper treatment of self-adjoint extensions of symmetric operators is required for a correct description of quantum systems. In particular, in this article we will focus on the definition of self-adjoint extensions of symmetric operators which are compatible with a given group of symmetries, $G$, of the dynamical system under investigation. This section is devoted to recall briefly the main results and definitions of the theory of self-adjoint extensions that will be needed, as well as to recall the notion of $G$-invariance.
\subsection{$G$-invariant self-adjoint extensions}\label{subsec:$G$-invariant self-adjoint extensions}
    Let $(\Omega,\eta)$ be a pair made up of a smooth manifold $\Omega$, with a smooth boundary $\partial \Omega$, and a Riemannian metric $\eta$. We will suppose that the measure $\mu$ is the Riemannian volume form associated with the Riemannian structure $\eta$. The boundary itself has the structure of a Riemannian smooth manifold without boundary, the metric tensor on the boundary, namely $\partial \eta$, being the pull-back of $\eta$ via the canonical inclusion $i: \partial \Omega \mapsto \Omega$. Let $\mathcal{H}={L}^2(\Omega,\mu)$ be the Hilbert space of square-integrable functions on $\Omega$ with respect to the measure $\mu$, and let $\mathbf{T}= -\Delta$ be the Laplace-Beltrami operator on $\Omega$ acting on a dense domain of the Hilbert space $\mathcal{H}$. Due to the presence of the boundary, the operator $\mathbf{T}$ defined on the subspace $C^{\infty}_0(\Omega)$ of smooth functions with compact support contained in the interior of $\Omega$, is an unbounded densely defined symmetric operator which in general is not self-adjoint: the specification of a proper set of boundary conditions allows to get different self-adjoint extensions $\mathbf{T}_b$. Even if it is not the most general one, the Laplace-Beltrami operator is a paradigmatic example and of fundamental importance in Quantum Mechanics. Indeed, it is the generator of the dynamics of a free particle on $\Omega$. Let us recall that an unbounded linear operator $\mathbf{T}$ on a complex, separable Hilbert space $\mathcal{H}$ with dense domain $\mathcal{D}(\mathbf{T})$ is symmetric if
    \begin{equation}
    \left\langle \Psi , \mathbf{T} \Phi \right\rangle = \left\langle \mathbf{T}\Psi , \Phi \right\rangle \quad \forall \Psi,\,\Phi \in \mathcal{D}(\mathbf{T}),
    \end{equation}
    and it is self-adjoint if, in addition, $\mathcal{D}(\mathbf{T}) = \mathcal{D}(\mathbf{T}^{\dagger})$, where $\mathbf{T}^{\dagger}$ is the adjoint operator whose domain is made up of those vectors $\Psi \in \mathcal{H}$ such that:
    \begin{equation}
    \left\langle \Psi, \mathbf{T}\Phi \right\rangle = \left\langle \chi, \Phi \right\rangle \quad \forall \Phi \in \mathcal{D}(\mathbf{T}),
    \end{equation}
    for some $\chi \in \mathcal{H}$. The set of vectors $\chi$ will form the range of the adjoint operator. A self-adjoint extension $\mathbf{T}_b$ of a symmetric operator $\mathbf{T}$ satisfies the conditions
    \begin{eqnarray}
    & \mathbf{T}_b \mid_{\mathcal{D}(\mathbf{T})} = \mathbf{T}\,,\\
    & \mathcal{D}(\mathbf{T}_b^{\dagger}) = \mathcal{D}(\mathbf{T}_b).
    \end{eqnarray}

    If the operator $\mathbf{T}$ is symmetric, but not essentially self-adjoint, different self-adjoint extensions determine different quantum systems. On the other hand, symmetries of a certain quantum system are implemented via unitary or anti-unitary linear operators acting on the Hilbert space $\mathcal{H}$ (cf. \cite{Wigner1931,Weyl1950}). Let $G$ be a group of symmetries of the dynamical quantum system described by the operator $\mathbf{T}$. This group of symmetries acts on $\mathcal{H}$ via the strongly continuous unitary representation $V: G \rightarrow \mathcal{U}(\mathcal{H})$. Therefore, it is worthwhile asking what self-adjoint extensions are compatible with the group action given by the map $V$, if we want the quantum system to have the group  of symmetries $G$. Following \cite{IbortLledoPerezPardo2015} we will consider the following definition of invariance for a given operator $\mathbf{T}$:
    \begin{Definition}\label{Def: Inv_op}
    Let $\mathbf{T}$ be a linear operator with dense domain $\mathcal{D}(\mathbf{T})\subset \mathcal{H}$ and consider a unitary representation $V: G \, \rightarrow \, \mathcal{U}(\mathcal{H})$. The operator $\mathbf{T}$ is said to be $G$-invariant if $\mathbf{T}V(g)\supseteq V(g)\mathbf{T}$, that is $V(g)\mathcal{D}(\mathbf{T}) \subset \mathcal{D}(\mathbf{T})$ for all $g\in G$ and
    \begin{equation}
        \mathbf{T}V(g)\Psi = V(g) \mathbf{T}\Psi \quad \forall g \in G, \: \forall \Psi \in \mathcal{D}(\mathbf{T}).
    \end{equation}
    \end{Definition}

    The issue of the existence and uniqueness of self-adjoint extensions of symmetric operators was addressed by von Neumann \cite{Neumann1930} who gave a characterisation of self-adjoint extensions in terms of unitary operators, $K$, between subspaces of $\mathcal{H}$, $(\mathcal{N}_+, \mathcal{N}_-)$, called deficiency spaces.

    However, if we think of self-adjoint extensions of differential operators on manifolds with a boundary (as we will consider in this work) the relationship between the unitary operator $K$ defining the self-adjoint extension $\mathbf{T}_K$ and the specification of a certain set of boundary conditions is hard to obtain, even if it is possible \cite{Grubb1968}. Therefore, for the Laplacian operator on a manifold ${\Omega}$ with a boundary $\partial {\Omega}$ an alternative approach has been proposed in \cite{AsoreyIbortMarmo2005, IbortLledoPerezPardo2015a}, according to which self-adjoint extensions of the operator $\mathbf{T}= -\Delta$ are in correspondence with unitary operators acting on the space $L^2(\partial \Omega)$ of square integrable functions on the boundary of $\Omega$. In the rest of the article we will address the issue of $G$-invariant self-adjoint extensions of the Laplace-Beltrami operator looking at the boundary data. Therefore, in the remainder of this section we will shortly outline some results concerning a large class of self-adjoint extensions of $\mathbf{T}$.

    Exploiting the relationship between quadratic forms and self-adjoint operators (c.f. \cite[Sec. VI.2]{Kato1995} for the details on this relation) it is possible to establish a correspondence between a class of self-adjoint extensions of the Laplace-Beltrami operator on a manifold $\Omega$ with a smooth boundary $\partial \Omega$ and a class of unitary operators on the space of square-integrable functions on its boundary (see \cite{IbortLledoPerezPardo2015a} for details). Before introducing this class of unitaries, it is worth providing some additional definitions. The Sobolev Hilbert space of order $k$ on a manifold $\Omega$ will be denoted by $H^k(\Omega)$ (see  for instance \cite{AdamsFournier2003} for a detailed presentation of Sobolev spaces).

    \begin{Definition}
    A unitary operator $U$ has spectral gap at $-1$ if one of the following conditions holds:
    \begin{enumerate}[label=\textit{(\roman*)},nosep]
        \item $\mathbbm{1}+U$ is invertible;\vspace*{0.5em}
        \item $-1$ belongs to the spectrum $\sigma(U)$ of the operator U but it is not an accumulation point of $\sigma(U)$.
    \end{enumerate}
    \end{Definition}
    Let $U$ be a unitary operator having $-1$ in its spectrum. If $P$ is the projector operator onto the eigenvalue $-1$ and $P^{\perp}=\mathbbm{1}-P$, the partial Cayley transform of the operator $U$ is defined as follows:
    \begin{equation}
        A_U\varphi = i P^{\perp}\frac{(\mathbbm{1}-U)}{(\mathbbm{1}+U)} \varphi,
        \quad \varphi \in L^2(\partial\Omega).
    \end{equation}

    \begin{Definition}
    Let $U$ be a unitary operator having spectral gap at $-1$. It is admissible if the partial Cayley transform $A_U$ leaves the subspace $H^{1/2}(\partial \Omega) $ invariant and is continuous with respect to the Sobolev norm of order $1/2$, i.e.
        \begin{equation}
            \norm{A_U\varphi}_{H^{1/2}(\partial \Omega)} \leq K \norm{\varphi}_{H^{1/2}(\partial \Omega)}, \quad \varphi \in H^{1/2}(\partial \Omega).
        \end{equation}
    \end{Definition}

    Then, there is a one-to-one correspondence between a class of self-adjoint extensions of the Laplace-Beltrami operator on a manifold $\Omega$ with smooth boundary and the class of unitary operators $U: L^2(\partial \Omega) \, \rightarrow \, L^2(\partial \Omega)$ which have spectral gap at $-1$ and are admissible.

    Given an admissible unitary with spectral gap at $-1$, the associated self-adjoint extension $\mathbf{T}_b$ is defined on the subspace $\mathcal{D}(\mathbf{T}_b)\subset \mathcal{D}(\mathbf{T}^{\dagger})\subset \mathcal{H}$
    whose elements $\Psi$ satisfy the following boundary conditions:
    \begin{equation}
    \psi - \mathrm{i}\dot{\psi} = U (\psi + \mathrm{i}\dot{\psi}),
    \end{equation}
    where $(\psi, \dot{\psi})$ are the boundary data, i.e. $\psi$ is the restriction to the boundary of the function $\Psi$ and $\dot{\psi}$ is the restriction of its normal derivative pointing outwards. These restrictions are obtained via the trace operator $\gamma : H^k(\Omega)\, \rightarrow \, H^{k-1/2}(\Omega)$, whenever $k > \frac{1}{2}$, which is a continuous surjective operator, with kernel defined by the functions in $H^k_0(\Omega)$ (see  for instance \cite{AdamsFournier2003} for more details).
    Having defined a class of self-adjoint operators in terms of unitary operators at the boundary, one can look for conditions on these unitaries such that the corresponding self-adjoint extensions are $G$-invariant with respect to some unitary representation of a group of symmetry $G$. Before introducing the main theorem of this section let us give some additional definitions.
    \begin{Definition}
        A strongly continuous unitary representation $V:G \, \rightarrow \, \mathcal{U}(L^2(\Omega))$ has a trace (or it is traceable) along the boundary $\partial \Omega$ if there exists another strongly continuous unitary representation $v: G \, \rightarrow \, \mathcal{U}(L^2(\partial \Omega))$ such that:
        \begin{equation}
            \gamma(V(g)\Psi) = v(g)\gamma(\Psi),
        \end{equation}
        for all $\Psi \in H^1(\Omega)$ and $g\in G $.
    \end{Definition}

    We can now state the main theorem characterising $G$-invariant self-adjoint extensions defined via admissible unitary operators at the boundary with spectral gap at $-1$:
    \begin{Theorem}[{\cite[Thm. 6.10]{IbortLledoPerezPardo2015a}}]\label{thm:$G$-invariant-extensions}
        Let $G$ be a group and $V: G \, \rightarrow \, \mathcal{U}(L^2(\Omega))$ a topological traceable representation of $G$, with unitary trace $v: G \, \rightarrow \, \mathcal{U}(L^2(\partial \Omega))$ along the boundary $\partial \Omega$. Let $U \in \mathcal{U}(L^2(\partial \Omega))$ be an admissible unitary operator with spectral gap at $-1$ defining a self-adjoint extension $\mathbf{T}_b$ of the Laplace-Beltrami operator. Assume that the self-adjoint extension corresponding to Neumann boundary conditions ($\dot{\psi}=0$) is $G$-invariant. Then we have the following cases:
        \begin{enumerate}[label=\textit{(\roman*)},nosep]
            \item If $\left[ v(g), U \right]=0$ for all $g \in G$ then $\mathbf{T}_b$ is $G$-invariant;\vspace*{0.5em}

            \item Consider the decomposition of the boundary Hilbert space $L^2(\partial \Omega) = W \oplus W^{\perp}$, where $W$ is the eigenspace relative to the eigenvalue $-1$ of the operator $U$, and denote by $P$ the orthogonal projection onto $W$. If $\mathbf{T}_b$ is $G$-invariant and $P\, :\, H^{1/2}(\partial \Omega) \, \rightarrow \, H^{1/2}(\partial \Omega) $ is continuous, then $\left[ v(g), U \right]=0$ for all $g\in G$.
        \end{enumerate}
    \end{Theorem}

    In the rest of the article we will make widely use of this result, since we will define self-adjoint extensions of the Laplace-Beltrami operator on $\Omega$ by specifying a suitable set of boundary conditions. Moreover, we will focus on one dimensional manifolds, and, as we will better clarify in the following sections, many of the technical requirements in the above discussion will be automatically satisfied.

\subsection{Additional remarks on symmetries and self-adjoint extensions}
    Before moving on, in this subsection we are going to add some remarks about the previous results.
    Let us start mentioning a result about $G$-invariant self-adjoint extensions in the von Neumann approach to the problem. Let $K \,: \, \mathcal{N}_+ \,\rightarrow \, \mathcal{N}_- $ be the unitary operator between the two deficiency spaces $\mathcal{N}_+, \,\mathcal{N}_-$ characterising a given self-adjoint extension $\mathbf{T}_K$ of an operator $\mathbf{T}$. Then $G$-invariant self-adjoint extensions can be built according to the result of the following theorem:
       \begin{Theorem}[{\cite[Thm. 3.5]{IbortLledoPerezPardo2015a}}]\label{thm:v-N G-invariant self-adjoint extensions}
        Let $\mathbf{T}: \mathcal{D}(\mathbf{T}) \subset \mathcal{H}\, \rightarrow \, \mathcal{H}$ be a closed, symmetric and $G$-invariant operator with equal deficiency indices. Let $\mathbf{T}_K$ be the self-adjoint extension defined by the unitary operator $K\in \mathcal{U}\left( \mathcal{N}_+, \mathcal{N}_- \right)$. Then $\mathbf{T}_K$ is $G$-invariant iff $V(g)K \xi = K V(g)\xi,\;\; \forall \xi \in \mathcal{N}_+,\: g\in G$.
    \end{Theorem}

    Therefore, if the symmetric operator $\mathbf{T}$ is $G$-invariant, $G$-invariant self-adjoint extensions are characterised by unitaries $K$ which commute with all the elements of the representation $V$ of the symmetry group $G$.
    However, one possibility not considered in \cite{IbortLledoPerezPardo2015a} is to study whether it is possible that the symmetric operator $\mathbf{T}$ is not $G$-invariant while one (or several) of its self-adjoint extensions $\mathbf{T}_b$ are. In particular, one could ask if there exist unitary representations $V$ of a symmetry group $G$ which do not preserve the domain $\mathcal{D}(\mathbf{T})$ of the operator $\mathbf{T}$ but preserve the domain of a self-adjoint extension $\mathbf{T}_b$. A simple instance is given by the following example (cf. \cite[Ch. 3]{AlbeverioGesztesyHoeghKrohnHolden2012}).

    \begin{Example}
        Let $\mathcal{H}={L}^2\left( \mathbb{R} \right)$ and let $\mathbf{T} = -\frac{\mathrm{d}^2}{\mathrm{d}x^2}$ with domain
        \[
            \mathcal{D}\left( \mathbf{T} \right) = \left\lbrace \Phi \in H^2\left( \mathbb{R} \right) \mid \varphi(0)=0  \right\rbrace.
        \]
        It is a symmetric operator with deficiency indices $(1,\,1)$, which implies that there are infinitely many self-adjoint extensions characterised by a single real parameter, say $\alpha$. The domains of these self-adjoint extensions, which we will call $\mathbf{T}_{\alpha}$, can be characterised as:
        \begin{equation}
            \mathcal{D}( \mathbf{T}_{\alpha} ) = \left\lbrace \Phi \in H^1 \left( \mathbb{R} \right) \cap H^2 \left( \mathbb{R} - \left\lbrace 0 \right\rbrace \right) \mid \varphi'(0_+) - \varphi'(0_-)= \alpha \varphi(0) \right\rbrace.
        \end{equation}
        The unitary representation of the group $\mathbb{R}$ given by
        \begin{equation}
            V(t) = \mathrm{e}^{-\mathrm{i}t\mathbf{T}_{\alpha}}
        \end{equation}
        preserves the domain $\mathcal{D}\left( \mathbf{T}_{\alpha} \right)$ and commutes with $\mathbf{T}_{\alpha}$ on its domain; therefore $\mathbf{T}_{\alpha}$ is $\mathbb{R}$-invariant. If $\alpha$ is positive the spectrum of $\mathbf{T}_{\alpha}$ is the positive half-line $\mathbb{R}_+$ and the generalised eigenfunctions (see, for instance, \cite[Chs. 4 and 9]{EspositoMarmoSudarshan2004}, \cite[Ch. 8]{ReedSimon2012} and \cite{gelfand1964generalised}) $\Psi_k(x)$ can be written in terms of a real parameter $k$ as follows:
        \begin{equation}
            \Psi_k(x) = \sqrt{\frac{4k^2}{2\pi \left( \alpha^2 + 4k^2 \right)}} \left( \mathrm{e}^{\mathrm{i}kx} + \frac{\mathrm{sign}(x) + \mathrm{sign}(k)}{2} \frac{\sin kx}{k} \right).
        \end{equation}
        Therefore the action of the unitary operator $V(t)$ can be expressed as follows:
        \begin{equation}
            V(t)\Phi(x) = \int_{\mathbb{R}}\mathrm{d}k \int_{\mathbb{R}}\mathrm{d}y \mathrm{e}^{-\mathrm{i}tk^2}\Psi_k(x)\bar{\Psi}_k(y) \Phi(y) = \int_{\mathbb{R}} \mathrm{d}y G(x,y;t)\Phi(y),
        \end{equation}
        where $G(x,y;t) = \int_{\mathbb{R}}\mathrm{d}k \mathrm{e}^{-\mathrm{i}tk^2}\Psi_k(x)\bar{\Psi}_k(y) $ is the Green-kernel associated with the operator $\mathbf{T}_{\alpha}$.

        In this case, however, the symmetric operator $\mathbf{T}$ is not $\mathbb{R}$-invariant since $\mathcal{D}( \mathbf{T} )$ is not invariant with respect to the group action $V(t)$. This is a consequence of the following Lemma (see \cite[Lemma 8.4.1]{MarsdenAbrahamRatiu1988}):

        \begin{Lemma}\label{lemma self-adjoint extensions}
            Let $\mathbf{T}_{\alpha}$ be an unbounded self-adjoint operator on a complex Hilbert space $\mathcal{H}$. Let $\mathcal{D}(\mathbf{T})\subset \mathcal{D}(\mathbf{T}_{\alpha})$ be a dense linear subspace of $\mathcal{H}$ and suppose $V(t) = \mathrm{e}^{-\mathrm{i}t\mathbf{T}_{\alpha}} $ leaves $\mathcal{D}(\mathbf{T})$ invariant. Then $\mathbf{T} = \mathbf{T}_{\alpha}\mid_{\mathcal{D}(\mathbf{T})}$ (i.e. $\mathbf{T}_{\alpha}$ restricted to $\mathcal{D}(\mathbf{T})$) is essentially self-adjoint.
        \end{Lemma}
        Let us recall that essentially self-adjoint means that the operator has a unique self-adjoint extension that coincides with its closure. Indeed, since in the case in question $\mathbf{T}$ is not essentially self-adjoint, the linear dense subspace $\mathcal{D}(\mathbf{T})$ cannot be invariant with respect to the action given by $V(t)$. Therefore, $\mathbf{T}_{\alpha}$ is $\mathbb{R}$-invariant whereas $\mathbf{T}$ is not.
    \end{Example}

    However, we know that it is possible to associate self-adjoint extensions of the Laplace-Beltrami operator on a manifold with boundary with unitaries on the space of boundary data. In this case, Theorem \ref{thm:$G$-invariant-extensions} was proven under two assumptions on the unitary representation $V$, i.e. it has to have a trace and preserve the Neumann self-adjoint extension. Therefore, one could see if there are $G$-invariant self-adjoint extensions with respect to a representation $V$ of a group $G$, while either the unitary representation is not traceable or it does not preserve the Neumann self-adjoint extension of the Laplace-Beltrami operator, or both. In fact, the system defined above is an example of this situation, too. Indeed, if a unitary representation is traceable, it has to preserve the kernel of the trace map. Since the kernel of the trace map, which, in the previous example, is given by the subspace $\mathrm{Ker}\gamma = \left\lbrace \Phi \in H^2\left( \mathbb{R} \right) \mid \varphi(0)=0,\: \dot{\varphi}(0)=0   \right\rbrace \subset L^2(\mathbb{R})$, is not preserved\footnote{Indeed, the operator $\mathbf{T}$ restricted to this domain is symmetric but not essentially self-adjoint. Therefore, according to Lemma \ref{lemma self-adjoint extensions}, this domain is not preserved by the representation V of the group $\mathbb{R}$.}, the unitary representation $V$ is not traceable. Therefore, concerning this situation, the assumption in Theorem \ref{thm:$G$-invariant-extensions} on the unitary representation, namely being traceable, is analogous to the assumption in Theorem \ref{thm:v-N G-invariant self-adjoint extensions} about the preservation of the domain of the symmetric operator. Since a careful analysis of symmetries which either are not traceable or are not compatible with the symmetric operator $\mathbf{T}$ or both, is out of the scope of this article, in the following sections we will limit to representations which will respect the hypotheses of the theorems introduced above, and we will postpone such an investigation to future works.

    After this introduction containing a critical review of previous results, in the following sections we will focus on the issue of self-adjoint extensions of the Laplace-Beltrami operator on quantum circuits, which are compatible with a given group of symmetries.

\section{Quantum Circuits: adjacency and local symmetries}\label{sec:local-symmetries}
    The central part of this work will focus on the description of self-adjoint extensions of the Laplace-Beltrami operator for one-dimensional manifolds. As we will see, this discussion can be fruitfully exploited in connection to the theory of quantum circuits. A quantum circuit is nothing but a quantum system consisting of a particle moving in a one-dimensional space. For such a system, the most general setting is that of a particle moving in a collection of intervals, $\Omega = \bigsqcup_{e \in \mathcal{E}} I_e$, arranged in such a way that some of them are connected through their endpoints. A good way to depict this arrangement is by associating a directed graph to $\Omega$, with edges representing the intervals $I_e$ and vertices representing their connections (see Figure \ref{fig:associated-graph}). We will use $\Omega$ to denote the union of intervals and $\mathcal{G}$ for the associated directed graph, whose vertex and edge spaces will be denoted by $\mathcal{V}$, $\mathcal{E}$ respectively.

    \begin{figure}[b]
        \begin{subfigure}[b]{0.45\textwidth}
            \centering
            \includegraphics[scale=1]{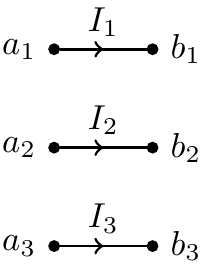}
            \caption{}\label{subfig:a}
        \end{subfigure}
        \begin{subfigure}[b]{0.45\textwidth}
            \centering
            \includegraphics[scale=1]{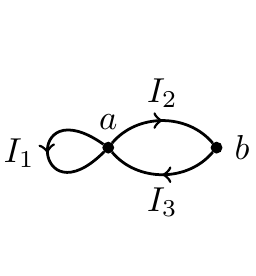}
            \caption{}\label{subfig:b}
        \end{subfigure}
        \caption{\small{Subfigure (A) shows the intervals $I_e$ for an $\Omega$ made out of three intervals. On Subfigure (B) we can see the associated graph if we connect on one side $a_1$ with $a_2$, $b_1$ and $b_3$, and in the other side $b_2$ with $a_3$. The first of the two connections is represented with the graph vertex labelled by $a$ and the second is represented with the vertex $b$.}} \label{fig:associated-graph}
    \end{figure}

    The Hilbert space associated with this quantum system is given by \hbox{$\mathcal{H} = L^2(\Omega)$} This means, in particular, that the Hilbert space can be characterised as:
    \[
        L^2(\Omega) = \left\{
            \Phi = (\Phi_1, \Phi_2, \dots, \Phi_n, \dots) \mid
            \Phi_e \in L^2(I_e),\, \sum_{e \in \mathcal{E}} \norm{\Phi_e}^2_{L^2(I_e)} < \infty
        \right\}.
    \]
    Analogously we have the Sobolev spaces $H^k(\Omega)$ with the norm defined as $\norm{\Phi}_{H^k(\Omega)} = \sum_{e \in \mathcal{E}} \norm{\Phi_e}^2_{H^k(I_e)}$. Note that since $\Omega$ is the disjoint union of intervals, there is no condition on the regularity of functions on $H^k(\Omega)$ at the endpoints: the function at each edge $\Phi_e$ is \emph{a priori} unrelated to other components of $\Phi = (\Phi_1, \dots, \Phi_n, \dots)$. On the contrary, the Sobolev space $H^k(\mathcal{G})$ would contain different regularity properties due to the topology of the digraph and the vertices connections.

    Let $a_e$, $b_e$ denote the endpoints of the interval $I_e = [a_e, b_e]$. It is clear that the boundary of $\Omega$ is the set $\partial \Omega = \{a_1, a_2, \dots, a_n, \dots, b_1, b_2, \dots, b_n, \dots\}$. Since the boundary consists of a countable set of points, whenever the trace map $\gamma$ is well defined, $\gamma(\Phi) = \Phi|_{\partial \Omega}$ will be a sequence of complex numbers which, in the most general case, does not need to converge. However, if we consider $\Phi \in H^1(\Omega)$ (and therefore every component $\Phi_e$ is continuous) the fact that both $\sum_{e \in \mathcal{E}} \norm{\Phi_e}^2_{L^2(I_e)}$ and $\sum_{e \in \mathcal{E}} \norm{\Phi_e'}^2_{L^2(I_e)}$ have to be finite make
    \[
        \varphi = \gamma(\Phi) = (\Phi_1(a_1), \Phi_1(b_1), \dots, \Phi_n(a_n), \Phi_n(b_n), \dots)
    \]
    a sequence in $\ell^2$, due to the continuity of the trace operator on each interval and provided that the the length of the intervals and their measures is uniformly bounded from above and from below. Moreover, if $\Phi \in H^2(\Omega)$, the same argument applied to $\Phi'$ implies that
    \[
        \dot{\varphi} = (-\Phi_1'(a_1), \Phi_1'(b_1), \dots, -\Phi_n'(a_n), \Phi_n'(b_n), \dots) \in \ell^2.
    \]
    Therefore the Hilbert space of boundary data coincides with the induced Hilbert space at the boundary of $\Omega$, $\ell^2$.

    Instead of the sequential order used above for the boundary data $\varphi$ and $\dot{\varphi}$, we can write them in a way more related to the topology of the underlying graph. We call this arrangement the \emph{threaded arrangement}, and it consists of placing together the values corresponding to connected endpoints (that is, corresponding to the same vertex on the associated graph). For instance, if we consider the example in Fig. \ref{fig:associated-graph} using the threaded arrangement we have
    \[\begin{alignedat}{2}
        \varphi &= \left(\Phi(a_1), \Phi(a_2), \Phi(b_1), \Phi(b_3), \Phi(a_3), \Phi(b_2)\right), \\
        \dot{\varphi} &= \left(-\Phi'(a_1), -\Phi'(a_2), \Phi'(b_1), \Phi'(b_3), -\Phi'(a_3), \Phi'(b_2)\right),
    \end{alignedat}\]
    where the first four component correspond to the vertex $a$ in the Fig. \ref{subfig:b} and the other two, to the vertex $b$.

    According to Subsection \ref{subsec:$G$-invariant self-adjoint extensions}, we can describe (a subset of) the self-adjoint extensions of the Laplacian on $\Omega$ in terms of admissible unitary operators $U: \ell^2 \to \ell^2$ with gap at $-1$ through the boundary equation
    \begin{equation}\label{eq:boundary-equation}
        \varphi - i \dot{\varphi} = U (\varphi + i \dot{\varphi}).
    \end{equation}
    At the end of the previous section we anticipated that the technical difficulties do not appear in the one-dimensional case, and it is due to the fact we established above: the space of boundary data is the same as the Hilbert space induced at the boundary. For this reason, every $U$ satisfies the admissibility condition. Hence we are left only with the spectral gap condition.

    According to Eq. \eqref{eq:boundary-equation}, a general unitary (with gap at $-1$) imposes general conditions involving the boundary data of all the components of functions on the domain of the associated self-adjoint extension. However, in order to model a quantum circuit its topology should be respected, in the sense that the probability associated with wave functions on the domain of the self-adjoint extension should be continuous throughout the circuit. Since continuity is a local property, it is natural to consider self-adjoint extensions associated to unitaries not relating boundary data from unconnected endpoints. Therefore, we are interested in self-adjoint extensions whose unitary, written according to the threaded arrangement, has a block structure given by the graph vertices:
    \[
        U = \bigoplus_{\nu \in \mathcal{V}} U_\nu.
    \]
    Further simplification follows when considering circuits whose graphs have finite degree on all its vertices. For such a system the blocks $U_\nu$ are $d_\nu \times d_\nu$ matrices, where $d_\nu$ is the degree of the vertex $\nu$, and the following property holds.

    \begin{Proposition}
        Let $\{U_i\}_{i \in \mathbb{N}}$ be unitary matrices and $U = \bigoplus_{i=1}^\infty U_i$ a unitary operator on $\ell^2$. Then the spectrum of $U$ is given by the closure of the union of the spectra of all $U_i$. That is, $\sigma(U) = \overline{\bigcup_{i=1}^\infty \sigma(U_i)}$.
    \end{Proposition}
    \begin{proof}
        Since $U$ is unitary, it is a bounded normal operator and thus cannot have residual spectrum (cf. \cite[Thm. 6.10.10]{NaylorSell1982}).

        Let $P_i$ be the orthogonal projector associated to the $i$-th block of $U$, that is $U P_i \varphi = P_i U \varphi = U_i P_i\varphi$ for every $\varphi \in \ell^2$.
        Since the blocks $U_i$ are finite-dimensional matrices, they have only point spectrum and given any eigenvector it is immediate to promote it to a vector $\varphi \in \ell^2$, extending by zero. Therefore, $\sigma(U_i) \subset \sigma(U)$ for every $i$.

        Suppose now that $\lambda \in \sigma(U)$. Since $\lambda$ is either in the point or in the continuous spectrum, there must exist a sequence $\varphi_n$ such that $\norm{\varphi_n} = 1$ and $\lim_{n \to \infty}\norm{(U - \lambda)\varphi_n} = 0$. One has that
        \[
            \norm{(U - \lambda) \varphi_n}^2 = \sum_{i=1}^\infty \norm{(U_i - \lambda)P_i\varphi_n}^2.
        \]
        Since $U_i$ is a unitary matrix, it is unitarily diagonalisable. Let $\{\mu_{ij}\}_{j=1}^{d_i}$ denote the eigenvalues of $U_i$ and $V_i$ the unitary such that $V_i^* U_i V_i = \diag(\mu_{ij})$. Then
        \[
        \begin{alignedat}{2}
            \norm{(U - \lambda) \varphi_n}^2
            &= \sum_{i=1}^\infty \sum_{j=1}^{d_i} |(\mu_{ij} - \lambda)(V_iP_i\varphi_n)_j|^2 \\
            &\geq \inf_{i \geq 1} \min_{1 \leq j \leq d_i} |\mu_{ij} - \lambda|^2
            \sum_{i=1}^\infty \sum_{j=1}^{d_i} |(V_iP_i\varphi_n)_j|^2 \\
            &= \inf_{i \geq 1} \min_{1 \leq j \leq d_i} |\mu_{ij} - \lambda|^2 \norm{\varphi_n}^2
        \end{alignedat}
        \]
        where $(V_iP_i\varphi_n)_j$ denotes the $j$-th component of the vector $V_iP_i\varphi_n$. Using that $\norm{\varphi_n} = 1$, one gets
        \[
            \norm{(U - \lambda) \varphi_n}^2 \geq \inf_{i \geq 1} \min_{1 \leq j \leq d_i} |\mu_{ij} - \lambda|^2.
        \]
        Therefore $\lambda \in \sigma(U)$ implies that $\inf_{\mu \in \bigcup \sigma(U_i)} |\mu - \lambda| = 0$, that is, $\lambda \in \overline{\bigcup \sigma(U_i)}$.
    \end{proof}

    \begin{Remark}
        This property has two immediate consequences we are going to use repeatedly on the rest of the work. First, if $U$ is made of a finite number of blocks which are repeated, then it has gap at $-1$. Second, if there is an $\varepsilon > 0$ such that for each block the lowest distance between an eigenvalue and $-1$ is greater than $\varepsilon$, then $U$ has gap at $-1$.
    \end{Remark}

    Since the probability associated to a wave function is given by its modulus, unitaries enforcing the continuity of such probability need to impose some conditions involving only the function trace, $\varphi$. It is not hard to see that conditions involving only $\varphi$ are given by the eigenvalue $-1$ of $U$. In fact, acting with the orthogonal projector over $\ker (U + 1)$, $P$, on both sides of \hbox{Eq. \eqref{eq:boundary-equation}} one gets
    \begin{equation} \label{eq:boundary-eq-projected}
        P \varphi = 0.
    \end{equation}
    Therefore, we are going to look into self-adjoint extensions such that \hbox{$\ker (U + 1)^{\perp}$} is the space of the traces $\varphi$ of functions such that $|\Phi|$ is continuous. It is clear that the projector $P$ must have the same block structure as the unitary $U$,
    \begin{equation}
        P = \bigoplus_{\nu \in \mathcal{V}} P_{\nu}.
    \end{equation}
    Let us denote by $\varphi_\nu$ the projection of $\varphi$ into the block associated with the vertex $\nu$. It is clear that for $|\Phi|$ to be continuous on the digraph $\mathcal{G}$, at every node the entries of $\varphi_\nu$ can differ only by a relative phase. The selection of the relative phases among the components of the vector $\phi_{\nu}$, amounts to $d_{\nu}-1$ conditions, in such a way that the space \hbox{$\ker(U_{\nu}+1)^{\perp}$} is the following one dimensional space:
    $$
        \ker(U_{\nu} + 1)^\perp = \linspan{(1, e^{i\alpha_1}, e^{i\alpha_2}, \dots, e^{i\alpha_{d_\nu-1}})^T}
    $$
    and thus $P_\nu^\perp = \mathbbm{1} - P_\nu$ has to be the rank-one orthogonal projector onto the linear span of $(1, e^{i\alpha_1}, e^{i\alpha_2}, \dots, e^{i\alpha_{d_\nu-1}})^T$. It immediately follows that the corresponding block of the unitary has the form
    \begin{equation}\label{eq:unitary-delta}
        U_\nu = e^{i\delta} P_\nu^\perp - (\mathbbm{1} - P_\nu^\perp),
    \end{equation}
    where $e^{i\delta}$ is the only eigenvalue of $U_\nu$ different from $-1$. We call this family of self-adjoint extension the quasi-$\delta$ family, characterised by the parameters $\delta, \alpha_1, \dots, \alpha_{d_\nu - 1}$. For all the self-adjoint extensions in this family, it is possible to restrict a suitably defined self-adjoint momentum operator. In other words, there exists a self-adjoint extension of the symmetric operator $P = -i \frac{\partial }{\partial x}$ which can be restricted to the domain of a quasi-$\delta$ self-adjoint extension of the Laplace-Beltrami operator, and is essentially self-adjoint on it.

    Now that the unitaries leading to a self-adjoint extension of the Laplacian compatible with a given quantum circuit have been characterised, for the rest of this section we focus on the study of the symmetries that are compatible with them. When applying Theorem \ref{thm:$G$-invariant-extensions} to the one-dimensional case, the fact that the space of traces is $\ell^2$ also makes that the extra condition on part \emph{(ii)} of the same theorem is always satisfied, and therefore both parts always hold. The condition for the self-adjoint extension to be invariant under the action of a group is that the trace representation commutes with the unitary operator defining the self-adjoint extension. Therefore, in the remainder of this work, we will characterise some unitary representations of a group of symmetries $G$ that commute with a matrix $U$ possessing a blockwise structure as in Eq. \ref{eq:unitary-delta}. For that, it is useful to consider each block separately leaving the global aspects for the next section.

    \begin{Definition}\label{def:local-symmetry}
        We say that the group $G$ is a local symmetry at the vertex $\nu$ of the self-adjoint extension determined by $U = \bigoplus_\nu U_\nu$ if it admits a faithful representation $v:G \to \mathcal{U}(d_\nu)$ such that $[v(g), U_\nu] = 0$ for every $g \in G$.
    \end{Definition}

    \begin{Remark}
        It is important to note that it is usually not possible to build an actual symmetry of the associated self-adjoint extension from a collection of local symmetries. In fact, the representation $v$ might not even be the trace of a representation in $L^2(\Omega)$. Because of this, local symmetries have to be understood only as tools to study global ones.
    \end{Remark}

    The following proposition provides a necessary and sufficient condition for $G$ to be a local symmetry of a self-adjoint extension compatible with the quantum circuit.

    \begin{Proposition}\label{prop:local-symmetries}
        Consider a quasi-$\delta$ self-adjoint extension with parameters $\delta, \alpha_1, \dots, \alpha_{d_\nu - 1}$ at the vertex $\nu$. A group $G$ is a local symmetry at $\nu$ if and only if it admits a unitary representation $v: G \to U(d_\nu)$ such that it leaves invariant the linear space $\mathrm{span}\left\lbrace (1, e^{i\alpha_1}, \dots, e^{i\alpha_{d_\nu - 1}})^T \right\rbrace$.
    \end{Proposition}
    \begin{proof}
        The proof follows straightforwardly from Eq. \eqref{eq:unitary-delta}. It is clear that $[v(g), U_\nu] = 0$ iff $[v(g),P_\nu^\perp] = 0$. Since $P_\nu^\perp$ is the orthogonal projector onto $\linspan{(1, e^{i\alpha_1}, \dots, e^{i\alpha_{d_\nu - 1}})^T}$, the sought result follows.
    \end{proof}

    Since the only condition to be a local symmetry of this family is to preserve a one-dimensional subspace of $\mathbb{C}^{d_\nu}$, it follows immediately the next corollary:

    \begin{Corollary}
        The biggest group which is a local symmetry at the vertex $\nu$ of a given quasi-$\delta$ family of self-adjoint extensions is the unitary group $U(d_\nu - 1)$.
    \end{Corollary}

    Before ending this section, let us remark that the given characterization of local symmetry at the vertex $\nu$ is nothing but the preservation of the structure of the eigenspaces of $U_\nu$. Consequently, the unitary group $U(d_\nu - 1)$ will not only preserve the family of quasi-$\delta$ self-adjoint extensions but also any other family with the same eigenspaces but different eigenvalues associated to each eigenspace. Each of this families can be seen as different orbits of a subgroup of $U(d_\nu)$ changing the invariant vector in Prop. \ref{prop:local-symmetries}:

    \begin{Proposition}\label{prop:quasi-delta-transformation}
        Let $\zeta = (1, e^{i\alpha_1}, \dots, e^{i\alpha_{d_\nu - 1}})^T$ and denote by $U_{\delta,\zeta}$ the quasi-$\delta$ self-adjoint extension parametrised by $\delta, \alpha_1, \dots, \alpha_{d_\nu}$. Then, every unitary matrix $V$ such that $V\zeta$ is proportional to $\zeta' = (1, e^{i\alpha'_1}, \dots, e^{i\alpha'_{d_\nu - 1}})^T$ transforms $U_{\delta, \zeta}$ into $U_{\delta, \zeta'}$. That is,
        \[
            V^* U_{\delta,\zeta} V = U_{\delta,\zeta'}.
        \]
    \end{Proposition}
    \begin{proof}
        Denote by $P_\zeta^\perp$ the orthogonal projector onto $\linspan{\zeta}$; since $V^*\zeta \propto \zeta'$, $V^* P_\zeta^\perp V = P_{\zeta'}^\perp$. Using this and Eq. \eqref{eq:unitary-delta} the result follows.
    \end{proof}

\section{Global symmetries on the Graph. $\mathbb{Z}$-invariance}\label{sec:global_symmetries}

    As pointed out in the introduction, in this section we are going to consider an example which is useful as an approximation of big, complex quantum circuits consisting of a repeated structure. The simplest example is the case of a graph which is the union of infinitely many intervals arranged in such a way to form a chain with loops at each node (see Figure \ref{fig:example-chain}). However, all the discussion can be extended straightforwardly to the general case of an arbitrary quantum circuit provided that the action of the considered group preserves the topology of the underlying graph.

    \begin{figure}[b]
        \centering
        \includegraphics[width=0.8\textwidth]{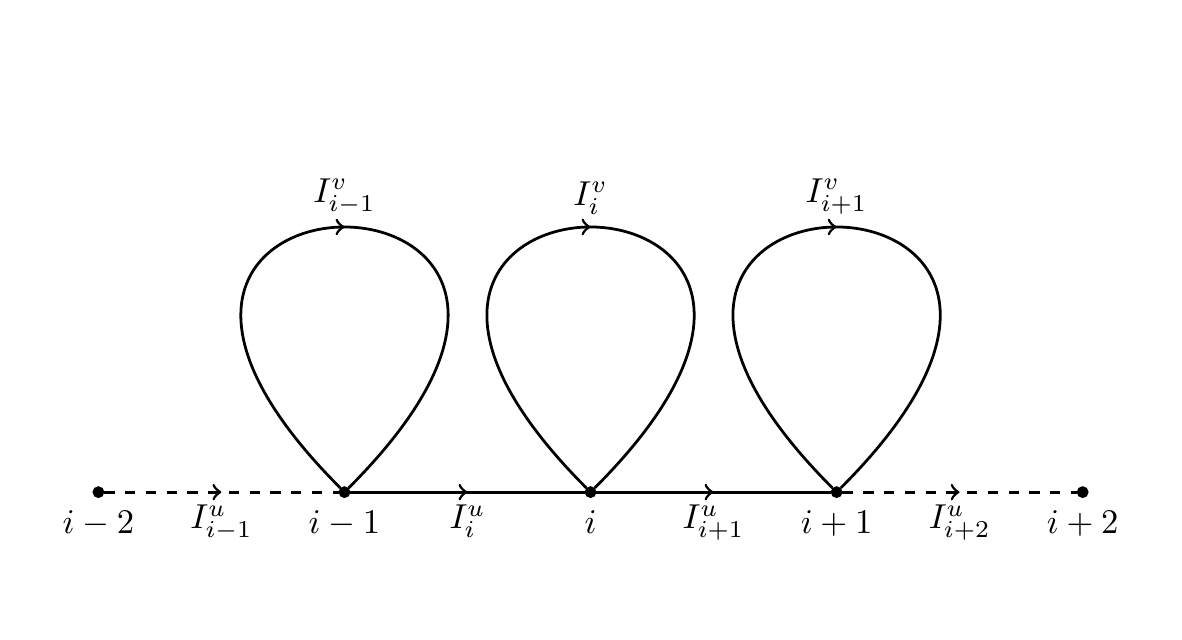}
        \caption{\small{Representation of the infinite graph associated with $\Omega$.}} \label{fig:example-chain}
    \end{figure}

    In order to properly describe such a structure let us fix the following nomenclature: any vertex is labelled by an integer $i\in \mathbb{Z}$. We can distinguish two families of edges in the graph: the edge connecting the vertices $i-1$ and $i$ will be denoted by $I_i^{u}$; the edge forming the loop at vertex $i$ will be called $I_i^{v}$. In this way it will be easy to generalise this periodic structure to a graph where the ``unit cells'' at each node will have a more complex structure.

    As already explained in the previous section, we consider the manifold
    \begin{equation}\label{manifold_periodic_graph}
        \Omega = \bigsqcup\limits_{\substack{i\in \mathbb{Z}\\ a\in \{ u,\,v \} } } I^a_i\,,
    \end{equation}
    which is a non-connected manifold made up of the disjoint union of all the pairs of intervals $\{ I_i^u,\,I_i^v \}$. A convenient way to describe this manifold is realised as follows. Assume that there is a reference interval $I = \left[ 0,\, 1 \right] $ and that there are diffeomorphisms
    \begin{eqnarray}
    &\xi^a_i \,:\, I\, \rightarrow\, I^a_i \nonumber \\
    &x \, \rightarrow \, \xi_i^a(x)
    \end{eqnarray}
    which play the role of local charts. Then, each point in $\Omega$ can be identified by this family of diffeomorphisms $\left\lbrace \xi^a_i \right\rbrace$, since for any point $P$ we can find a point $x\in I$ and a diffeomorphism, say $\xi^a_i$, such that $P=\xi_i^a(x)$.

    Each interval $I_i^a$ has the structure of a Riemannian manifold with metric $\eta_i^a$. Consequently the whole manifold $\Omega$ possesses a Riemannian metric tensor field
    $$\eta = \bigoplus\limits_{\substack{i \in \mathbb{Z}\\ a\in \{ u,\,v \}}} \eta_i^a \,.$$
    A basis of the tangent space $T_P\Omega$ at the point $P\in \Omega$ is made up of one vector which is written $\frac{\partial }{\partial x}$ in the coordinate $x$. A generic vector field $X \in \mathfrak{X}(\Omega)$ on the manifold $\Omega$ will be written as
    $$
    X = \bigoplus\limits_{\substack{i \in \mathbb{Z}\\ a\in \{ u,\,v \}}} X_i^a
    $$
    and we have that $X(P) = X_i^a(P) = X^a_i( \xi_i^a(x)) \frac{\partial }{\partial x}$ where we have named $X^i_a$ both the vector field and its unique component in a given basis.

    \noindent Therefore, the action of the metric tensor field $\eta$ on a real valued vector field $X$ can be expressed as follows:
    \begin{equation}
        \eta (X ,\, X) (\xi_i^a(x)) = \eta_i^a\left(\xi_i^a (x) \right) \left(X_i^a(\xi_i^a(x))\right)^2, \quad x \in I.
    \end{equation}
    The corresponding invariant Riemannian volume form $\mathrm{d}\mu \in \Lambda(\Omega)$ is given by
    \begin{equation}
    \mathrm{d}\mu (P) = \sqrt{\eta_i^a(\xi_i^a(x))} \mathrm{d}x\,.
    \end{equation}
    The previous definitions lead to the following Laplace-Beltrami operator
    \begin{equation}
    \Delta = \bigoplus\limits_{\substack{i \in \mathbb{Z}\\ a\in \{ u,\,v \}}} \Delta_i^a
    \end{equation}
    which acts on a dense subspace of the space of square integrable functions over $\Omega$, denoting the latter by $L^2(\Omega)$ with the induced scalar product defined by
    \[
        \langle \Psi, \Phi \rangle = \sum\limits_{\substack{i \in \mathbb{Z}\\ a\in \{ u,\,v \}}} \langle \Psi_i^a, \,\Phi_i^a \rangle_{i,a}.
    \]
    Each of the ``local'' Laplace-Beltrami operators can be written as follows:
    \begin{equation}
        (\Delta\Psi)(\xi_i^a(x)) = \left(\Delta_i^a \Psi_i^a \right)(x) = -\frac{1}{\sqrt{\eta_i^a(x)}} \frac{\partial }{\partial x} \left( \frac{1}{\sqrt{\eta_i^a(x)}} \frac{\partial}{\partial x}\Psi_i^a(x) \right)
    \end{equation}
    where $\Psi_i^a(x)$ is the function in $L^2(I,\mathrm{d}x^i_a)$, where the measure $dx^i_a$ is the pull-back of the measure $d\mu_i^a$ via the diffeomorphism $\xi_i^a$. This function represents $\Psi_i^a \in L^2(I_i^a, \mathrm{d}\mu_i^a)$, and it is supposed regular enough for the previous equation to make sense.

    Let us consider the following symmetric domain for the Laplace-Beltrami operator
    \begin{equation}
    \mathcal{D}_0 = C_c^{\infty}(\Omega)\,,
    \end{equation}
    which is the space of smooth functions with compact support contained in the interior of $\Omega = \bigsqcup_{\tiny{i\in \mathbb{Z}, a\in \{ u,\,v \} } } I_i^a$. This implies that these functions can have support only on a finite number of intervals. Correspondingly, the domain of the adjoint operator is the Sobolev space $H^2(\Omega)$. As already explained in the previous section, the space of self-adjoint extensions of this operator compatible with the topology of the graph is in correspondence with the space of unitaries $U$ which have a blockwise structure, i.e.
    $$
    U=\bigoplus_{i \in \mathbb{Z}} U_i\,,
    $$
    each block $U_i$ being a $4\times 4$ unitary matrix representing the node of the form showed in Fig. \ref{fig:example-chain-node}.

    \begin{figure}[b]
        \centering
        \includegraphics[width=0.4\textwidth]{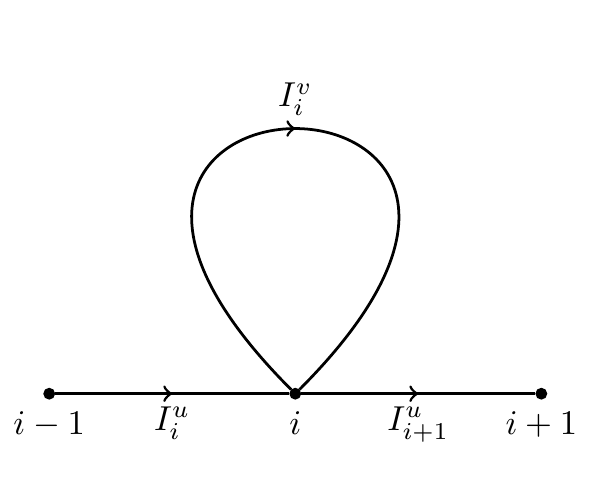}
        \caption{\small{Elementary cell of the graph $\Omega$.}} \label{fig:example-chain-node}
    \end{figure}

    Any unitary $U_i$ represents the ``local'' boundary conditions at each vertex of the graph via the formula:
    \begin{equation}
        \varphi_i - i \dot{\varphi}_i = U_i \left( \varphi_i + i \dot{\varphi}_i \right),
    \end{equation}
    where $\varphi_i$ and $\dot{\varphi}_i$ denote the boundary data at the node labelled by $i$:
    \begin{equation}
        \varphi_i= \left(
        \begin{array}{c}
            \Phi_i^u(1) \\
            \Phi_i^v(0)\\
            \Phi_i^v(1)\\
            \Phi_{i+1}^u(0)
        \end{array}
        \right),
        \qquad
        \dot{\varphi}_i = \left(
        \begin{array}{c}
            (\Phi_i^u)'(1)\\
            -(\Phi_i^v)'(0)\\
            (\Phi_i^v)'(1)\\
            -(\Phi_{i+1}^u)'(0)
        \end{array}
        \right).
    \end{equation}

    We now want to study symmetry transformations on our graph, such that we can apply the results of section 2. In particular we will consider an action of the group $\mathbb{Z}$, and we will start with an action which preserves both the boundary and the adjacency properties of the nodes. We define therefore the following map:
    \begin{eqnarray}
        & \mathfrak{V}\,:\,\mathbb{Z}\times \Omega\, \rightarrow \, \Omega \nonumber\\
        &(k, \,P)\, \mapsto\, Q=\mathfrak{V}_k (P)
    \end{eqnarray}
    such that, if $P\in I_i^a$ with $P=\xi_i^a(x)$, then $Q\in I_{i+k}^a$ and
    \begin{equation}
    Q = \xi_{i+k}^a(x)=\mathfrak{V}_k\xi_i^a(x)\,.
    \end{equation}
    This action, for a fixed $k\in \mathbb{Z}$, is implemented by a family of diffeomorphisms $\mathfrak{V}^a_{i+k,i}$
    \begin{eqnarray}
    & \mathfrak{V}^a_{i+k,i}\,:\,I_i^a\, \rightarrow \, I_{i+k}^a \nonumber\\
    &P\mapsto Q = \xi_{i+k}^a\circ (\xi_i^a)^{-1} (P) := \mathfrak{V}^a_{i+k,i}(P)\,,
    \end{eqnarray}
    which preserve the boundary.
    It induces an action on the space $\mathcal{D}_0 \subset L^2(\Omega)$ of smooth functions with compact support on $\Omega$, via the pull-back, i.e.
    \begin{eqnarray}
    & V\,:\,\mathbb{Z} \, \rightarrow \, \mathcal{B}\left( L^2(\Omega) \right) \nonumber \\
    & k \mapsto (V_k \Phi)(\xi_i^a(x)) = \left( \mathfrak{V}_{-k}^*\Phi \right) (\xi_i^a(x)) = \Phi_{i-k}^a(\xi_{i-k}^a(x))\,.
    \end{eqnarray}
    In addition, we want this action to define a symmetry of the Laplace-Beltrami operator over $\Omega$ with domain $\mathcal{D}_0$, i.e. for all $k \in \mathbb{Z}$
    \begin{equation}
    \Delta V_k(\Phi) = V_k(\Delta \Phi)\,, \quad \Phi \in \mathcal{D}_0\,.
    \end{equation}
    Clearly $V_k \mathcal{D}_0 \subset \mathcal{D}_0$, for all $k\in \mathbb{Z}$. A sufficient condition for this action to implement a unitary representation of $\mathbb{Z}$ such that $\Delta$ is $\mathbb{Z}$-invariant is that $\mathfrak{V}_k$ acts as an isometric diffeomorphism, i.e. $\mathfrak{V}_k^*\eta = \eta$ for all $k \in \mathbb{Z}$. This requirement imposes a restriction on the class of metric tensors that one can consider; indeed, on one hand we have
    \[
        \mathfrak{V}_k^*\eta \left( X , Y \right)(\xi_{i}^a(x))
        = \eta \left( (\mathfrak{V}_k)_* X, (\mathfrak{V}_k)_* Y \right) \left( \xi_{i+k}^a(x)\right)
        = \eta_{i+k}^a(\xi_{i+k}^a(x)) X^a_i Y^a_i\,,
    \]
    where $X_i^a,Y_i^a$ are the components of two tangent vectors at the interval $I_i^a$, and they remain invariant under the action of the push-forward. On the other hand
    \[
        \eta \left( X, Y \right)(\xi_{i}^a(x)) = \eta_{i}^a(\xi_i^a(x)) X^a_i Y^a_i.
    \]
    Consequently if $\eta$ must be invariant, $\eta_{i+k}^a(x) = \eta_i^a(x)$ for all $k\in \mathbb{Z}$. Hence, one has that $I_i^a$ and $I_j^a$ are isometric for all $i,j \in \mathbb{Z}$. Since we are in one dimension, one can always choose the diffeomorphism $\xi_i^a$ such that the metric tensor is constant, and the only remaining freedom in order to choose the Riemannian structure of the manifold $\Omega$ amounts to fix the length of the intervals of type $u$ and $v$ for all $i\in \mathbb{Z}$. We shall denote this lengths $l_u,\,l_v$, respectively.

    Since the group $G$ acts on $\Omega$ via isometric diffeomorphisms, the corresponding action $V$ on the space $L^2(\Omega)$ is a unitary representation of the group $\mathbb{Z}$ which is traceable. Concerning this last point, in fact, it is enough to consider the restriction of the action of $G$ to the boundary $\partial \Omega$:
    \begin{eqnarray}
    & \mathfrak{v}_k \,: \, \partial \Omega \, \rightarrow \, \partial \Omega \\
    & P \, \rightarrow Q = \mathfrak{v}_k (\xi_i^a(x)) = \mathfrak{V}_k (\xi_i^a(x))\,, \quad x\in \{ 0,1 \} \,,
    \end{eqnarray}
    because, as already mentioned, this action preserves the boundary. Then, the trace representation $v\,:\, G \, \rightarrow \, \mathcal{U}(L^2(\partial \Omega))$ is obtained from $\mathfrak{v}_k$ in the same way as $V$ is obtained from $\mathfrak{V}_k$, i.e., via the pull-back.

    This is not the unique traceable representation that one can obtain from the action $\mathfrak{V}_k$. Since $\mathbb{Z}$ is an infinite cyclic group, any unitary representation of the group $G$ can be obtained from a unitary representation of its generator $\mathfrak{V}_1$. Therefore, we can consider more generally the unitary representation given by
    \begin{eqnarray} \label{eq:V-commutes-with-laplacian}
        & V^{\theta}\,:\,\mathbb{Z} \, \rightarrow \, \mathcal{U}(L^2(\Omega))\nonumber \\
        & (V^{\theta}_1\Phi)(\xi_i^a(x)) = \mathrm{e}^{-i\theta(\xi_i^a(x))} \left( \mathfrak{V}_{-1}^*\Phi \right)(\xi_i^a(x)) = \nonumber \\
        & = \mathrm{e}^{-i \theta(\xi_i^a(x))}\Phi_{i-1}^a(\xi_{i-1}^a(x))\,,
    \end{eqnarray}
    where $\theta(\xi_i^a(x))$ can depend both on the interval and the point on the interval. However, for the case under investigation, this action commutes with the Laplace-Beltrami operator on an interval $I_i^a$ if $\theta^i_a$ is constant on each interval.
    Indeed, we have that
    \begin{equation}
        \Delta(\mathrm{e}^{-i\alpha}\Phi) = \Delta(\mathrm{e}^{-i\alpha}) \Phi + 2 \frac{\partial \mathrm{e}^{-i\alpha}}{\partial x} \frac{\partial \Phi}{\partial x} + \mathrm{e}^{-i\alpha}\Delta(\Phi)
    \end{equation}
    and the multiplication by a constant phase commutes with the Laplace-Beltrami operator.

    \noindent This representation of the generator induces the following unitary representation of the cyclic group $\mathbb{Z}$:
    \begin{eqnarray}\label{eq:representation}
        & V^{\theta}\,:\, \mathbb{Z} \, \rightarrow \, \mathcal{U}(L^2(\Omega))\\
        & k \mapsto \left( V_k^{\theta}\Phi \right) (\xi_i^a(x)) = \mathrm{e}^{-i\sum_{n=1}^{k}\theta^a_{i+n}}(V_k\Phi)(\xi_i^a(x))
    \end{eqnarray}
    and each phase $\theta_j^a$ is defined up to the sum of a multiple of $2\pi$.

    In particular, if the function $\theta_j^a$ is the same for each $j\in \mathbb{Z}$ we get the following representation
    \begin{eqnarray}
        & V^{\theta}\,:\, \mathbb{Z} \, \rightarrow \, \mathcal{U}(L^2(\Omega))\\
        & k \, \rightarrow \left( V_k^{\theta}\Phi \right) (\xi_i^a(x)) = \mathrm{e}^{-ik\theta^a}(V_k\Phi)(\xi_j^a(x))\,,
    \end{eqnarray}
    where two different choices for the parameter $\theta^a$, satisfying $\theta^a - \tilde{\theta}^a = 2m\pi$ for some $m\in \mathbb{Z}$, define the same representation of the group $\mathbb{Z}$.

    The representation above induces a trace representation which acts in the following way:
    \begin{equation}\label{eq:trace-representation}
        (v_k \varphi)_i^a = \mathrm{e}^{-i\sum_{n=1}^k \theta_{i+n}^a} \varphi_{i-k}^a.
    \end{equation}

    Before we state the main result of this section, let us recall that the quasi-$\delta$ family of self-adjoint extensions is characterised by the parameters $\{\delta_i, \alpha_1^i, \dots, \alpha_{d_i}^i\}_{i \in \mathbb{Z}}$, determining the unitary block $U_i$ associated with the $i$-th vertex.

    \begin{Theorem}\label{thm:Z-inv}
        Consider the quantum circuit depicted on Figure \ref{fig:example-chain}. The group $\mathbb{Z}$ is a symmetry compatible with every quasi-$\delta$ self-adjoint extension such that:
        \begin{enumerate}[label=\textit{(\roman*)},nosep]
            \item The value of the parameter $\delta_i$ is the same at each vertex, $\delta$.\vspace*{0.5em}
            \item The relative phases between the boundary data associated with the loop is the same in all the vertices, i.e. $\alpha_1^i - \alpha_2^i$ does not depend on the vertex $i$.
        \end{enumerate}
    \end{Theorem}
    \begin{proof}
        Consider the representation given in Eq. \eqref{eq:representation}. Due to the fact that $\theta$ does not depend on $x$ and $V^\theta_k$ is unitary for all $k$, it follows that this representation preserves the Neumann self-adjoint extension (see \cite[Prop. 6.14]{IbortLledoPerezPardo2015}). Then it follows that the requirements of Theorem \ref{thm:$G$-invariant-extensions} are fulfilled, and we only need to proof that there exist $\theta_i^a$ such that the trace action defined on Eq. \eqref{eq:trace-representation} satisfies $v_k^* U v_k = U$. Since the group is cyclic, it is enough to prove it for its generator, $v_1$.
        The matrix representation of $U$ and $v_1$ is the following:
        \begin{equation*}
        U=\left(
        \begin{array}{ccccc}
        \cdots & 0 & 0 & 0 & \cdots \\
        \cdots & U_{i-1} & 0 & 0 & \cdots \\
        \cdots & 0 & U_i & 0 & \cdots \\
        \cdots & 0 & 0 & U_{i+1} & \cdots \\
        \cdots & 0 & 0 & 0 & \cdots \\
        \end{array}
        \right)
        \quad
        v_1= \left(
        \begin{array}{ccccc}
        \cdots & v^{i-2}_1 & 0 & 0 & \cdots \\
        \cdots & 0 & v_1^{i-1} & 0 & \cdots \\
        \cdots & 0 & 0 & v_1^{i} & \cdots \\
        \cdots & 0 & 0 & 0 & v_1^{i+1} \\
        \cdots & 0 & 0 & 0 & \cdots \\
        \end{array}
        \right)
        \end{equation*}
        where each element in the matrix is a $4\times 4$ block. From its definition, it is clear that $v_1$ respects the block structure of $U$ and that the $i$-th block of $v_1^* U v_1$ is equal to $(v_1^{i-1})^*U_{i-1}v_1^{i-1}$, where $v_1^{i-1} = \diag(\mathrm{e}^{-i\theta_{i-1}^u}, \mathrm{e}^{-i\theta_i^v}, \mathrm{e}^{-i\theta_i^v}, \mathrm{e}^{-i\theta_{i}^u})$. Since $v_1^i$ is unitary, $(v_1^{i-1})^*U_{i-1}v_1^{i-1}$ has the same eigenvalues as $U_{i-1}$, and $\delta_i$ must be equal to $\delta_{i-1}$. This is satisfied by hypothesis \emph{(i)}.

        The only thing left to show is that $U_i = (v_1^{i-1})^*U_{i-1}v_1^{i-1}$. By Prop. \ref{prop:quasi-delta-transformation}, this is equivalent to $v_1^{i-1} \zeta_{i-1} \propto \zeta_i$, where $\zeta_i = (1, \mathrm{e}^{i\alpha_1^i}, \mathrm{e}^{i\alpha_2^i}, \mathrm{e}^{i\alpha_3^i})^T$. One has that
        \[
            v_1^{i-1} \zeta_{i-1}= \begin{pmatrix}
                \mathrm{e}^{-i\theta_{i-1}^u} \\
                \mathrm{e}^{-i(\theta_i^v - \alpha_1^{i-1})} \\
                \mathrm{e}^{-i(\theta_i^v - \alpha_2^{i-1})} \\
                \mathrm{e}^{-i(\theta_i^u - \alpha_3^{i-1})}
            \end{pmatrix}, \qquad
            \zeta_i = \begin{pmatrix}
                1 \\
                \mathrm{e}^{i\alpha_1^i} \\
                \mathrm{e}^{i\alpha_2^i} \\
                \mathrm{e}^{i\alpha_3^i}
            \end{pmatrix},
        \]
        which shows that $\theta_i^a$ can be chosen so that $v_1^{i-1} \zeta_{i-1} \propto \zeta_i$ if condition \emph{(ii)} holds.
    \end{proof}

    Before concluding this section, it is important to add a remark. Indeed, the contents of the previous theorem can be analysed also from a different perspective: for any representation $V^{\theta}$ of the group $\mathbb{Z}$, there is a family of quasi-$\delta$ self-adjoint extensions compatible with it. In particular, the action $V^\theta$ with $\theta_i^a$ constantly zero is compatible only with self-adjoint extensions characterised by $\zeta_i=(1,\mathrm{e}^{i\alpha_2}, \mathrm{e}^{i\alpha_2}, \mathrm{e}^{i\alpha_3})^T$, which amounts to require that the quasi-$\delta$ condition is the same on every vertex. If $\theta_i^a$ is constant and does not depend on the interval $i$ all the relative phases can be incremented by a constant factor from one vertex to another, i.e. $\alpha_l^i - \alpha_l^{i+1}$ must remain constant, and also $\alpha_1^i - \alpha_2^i$ needs to be independent from $i$ by hypothesis \emph{(ii)}. Whereas, only if the phase $\theta_j^a$ in the representation $V^{\theta}$ can vary from an interval to the other, self-adjoint extensions with $\alpha_l^i$ varying from a vertex to the other are $\mathbb{Z}$-invariant, once more with the restriction that $\alpha_1^i - \alpha_2^i$ must be vertex independent. In each of these situations we can compute the spectrum of the Laplace-Beltrami operator. Indeed, after computing the solutions of the equation $\Delta_i^a \Phi_i^a = k^2 \Phi_i^a$ in each interval, the generalised eigenfuctions of the Laplace-Beltrami operator can be obtained after imposing the chosen boundary conditions. This procedure will result in a set of algebraic linear equations involving the boundary data. For the family of self-adjoint extensions considered in Theorem \ref{thm:Z-inv}, if $k$ is an eigenvalue and $\Phi_i^a(x)=A_i^a \mathrm{e}^{ikx}+B_i^a \mathrm{e}^{-ikx}$ the corresponding solution on the interval $I_i^a$, at the $i$-node one gets the following system of equations:
    \begin{eqnarray}
    & A_i^{u} \mathrm{e}^{ik} + B_i^u \mathrm{e}^{-ik} = \mathrm{e}^{i \alpha_1^i} \left( A_i^{v}  + B_i^v \right) \label{general boundary conditions: eq.1} \\
    & A_i^{u} \mathrm{e}^{ik} + B_i^u \mathrm{e}^{-ik} = \mathrm{e}^{i\alpha_2^i}\left( A_i^{v} \mathrm{e}^{ik} + B_i^v \mathrm{e}^{-ik} \right) \label{general boundary conditions: eq.2} \\
    & A_i^{u} \mathrm{e}^{ik} + B_i^u \mathrm{e}^{-ik} = \mathrm{e}^{i\alpha_3^i} \left( A_{i+1}^{u}  + B_{i+1}^u \right) \label{general boundary conditions: eq.3} \\
    &4\tan\left( \frac{\delta}{2} \right) \left( A_i^{u} \mathrm{e}^{ik} + B_i^u \mathrm{e}^{-ik} \right) = ik \left( \left( A_i^{u} \mathrm{e}^{ik} - B_i^u \mathrm{e}^{-ik} \right) +\right. \nonumber\\
    & \left. - \mathrm{e}^{i\alpha_1^i} \left( A_i^{v} - B_i^v  \right) + \mathrm{e}^{i\alpha_2^i} \left( A_i^{v} \mathrm{e}^{ik} - B_i^v \mathrm{e}^{-ik} \right) - \mathrm{e}^{i\alpha_3^i} \left( A_{i+1}^{u} - B_{i+1}^u \right) \right) \label{general boundary conditions: eq.4} \,,
    \end{eqnarray}
    where $\alpha_2^i-\alpha_1^i = \alpha$ is constant along the graph and $\delta$ as well. As a particular instance of this general form, we can consider the self-adjoint extension for which $\alpha^i_1=\alpha^i_2=\alpha^i_3=\delta=0$ for all $i \in \mathbb{Z}$. In this case the final system of equations simplifies and we get:
    \begin{eqnarray*}
    &A_i^{u} \mathrm{e}^{ik} + B_i^u \mathrm{e}^{-ik} = A_i^{v} + B_i^v \\
    &A_i^{v} \mathrm{e}^{ik} + B_i^v \mathrm{e}^{-ik} = A_i^{v} + B_i^v \\
    &A_i^{u} \mathrm{e}^{ik} + B_i^u \mathrm{e}^{-ik} = A_{i+1}^{u} + B_{i+1}^u \\
    &A_i^{u} \mathrm{e}^{ik} - B_i^u \mathrm{e}^{-ik} - \left( A_i^v - B_i^v \right) + A_i^{v} \mathrm{e}^{ik} - B_i^v \mathrm{e}^{-ik} - \left( A_{i+1}^{u}  - B_{i+1}^u \right) = 0
    \end{eqnarray*}
    If we introduce the new variables
    \begin{eqnarray*}
    & A^{(1)} = A^v_i+A_i^u \\
    & A^{(0)} = A^v_i+A_{i+1}^u \\
    & B^{(1)} = B^v_i+B_i^u \\
    & B^{(0)} = B^v_i+B_{i+1}^u\,,
    \end{eqnarray*}
    it is immediate to derive the equation $A^{(0)}= A^{(1)}\mathrm{e}^{ik}$ and $B^{(0)}=B^{(1)}\mathrm{e}^{-ik}$. From these conditions we derive the following relationships:
    \begin{eqnarray*}
    &A_i^u-A_{i+1}^u = A^{(1)}(1-\mathrm{e}^{ik})\\
    &B_i^u-B_{i+1}^u = B^{(1)}(1-\mathrm{e}^{-ik})
    \end{eqnarray*}
    After some straightforward computations, one can express the value of the coefficients $\{ A_i^u,B^u_i \}$ in terms of $\{ A^{(1)},B^{(1)} \}$ as follows:
    \begin{eqnarray}
    & 2A_i^u \left( 1+ \mathrm{e}^{ik} \right) = A^{(1)}\left( 2+\mathrm{e}^{ik} \right) - B^{(1)}\mathrm{e}^{-ik} \\
    & 2B_i^u \left( 1+ \mathrm{e}^{ik} \right) = B^{(1)}\left( 1+2\mathrm{e}^{ik} \right) - A^{(1)}\mathrm{e}^{2ik}\,.
    \end{eqnarray}

    The same strategy can be used to find the solutions of the equations \eqref{general boundary conditions: eq.1}-\eqref{general boundary conditions: eq.4} for other values of the parameters. Indeed we can introduce the new variables
    \begin{eqnarray*}
        & A_{out}^{i} = A^v_i\mathrm{e}^{i\alpha^i_2} +A_i^u \\
        & A_{in}^{i} = A^v_i\mathrm{e}^{i\alpha^i_1} +A_{i+1}^u\mathrm{e}^{i\alpha_3^i} \\
        & B_{out}^{i} = B^v_i\mathrm{e}^{i\alpha^i_2} +B_i^u \\
        & B_{in}^{i} = B^v_i\mathrm{e}^{i\alpha^i_1} +B_{i+1}^u\mathrm{e}^{i\alpha_3^i}\,,
    \end{eqnarray*}
    which, for $\delta=0$ satisfy the relation\footnote{If $\delta \neq 0$ the relation between $\left\lbrace A_{out}^i, B_{out}^i \right\rbrace$ and $\left\lbrace A_{in}^i, B_{in}^i \right\rbrace$ is more complicated, but using it we can solve the initial system following the same steps illustrated for the case $\delta = 0$.} $A_{in}^i = A_{out}^i \mathrm{e}^{ik}$ and $B_{in}^i = B_{out}^i \mathrm{e}^{-ik}$. From the definition of the new variables and their relations, we get the following expressions:
    \begin{eqnarray*}
    &A_i^u-A_{i+1}^u \mathrm{e}^{i(\alpha_3^i+\alpha)} = A^{i}_{out}(1-\mathrm{e}^{i(k+\alpha)})\\
    &B_i^u-B_{i+1}^u \mathrm{e}^{i(\alpha_3^i+\alpha)} = B^{i}_{out}(1-\mathrm{e}^{i(\alpha-k)})\\
    &A_i^u(\mathrm{e}^{-i\alpha} + \mathrm{e}^{ik}) + B_i^u(\mathrm{e}^{-i\alpha} + \mathrm{e}^{-ik}) = (A^i_{out}+B^i_{out})\mathrm{e}^{-i\alpha}\\
    & A_i^{u} \mathrm{e}^{ik} + B_i^u \mathrm{e}^{-ik} = \mathrm{e}^{i\alpha_3^i} \left( A_{i+1}^{u}  + B_{i+1}^u \right)\,,
    \end{eqnarray*}
    and straightforward substitutions lead to the following results:
    \begin{eqnarray}
    & 2A_i^u \left(\mathrm{e}^{2ik} -1 \right) = A_{out}^i\left( \mathrm{e}^{2ik} + \mathrm{e}^{i(k+\alpha)} - 2 \right) + B_{out}^i (\mathrm{e}^{-i(k-\alpha)}-1 )\\
    & 2B_i^u \left(1 - \mathrm{e}^{-2ik} \right) = B_{out}^i\left( 2 - \mathrm{e}^{-2ik} - \mathrm{e}^{-i(k-\alpha)} \right) + A_{out}^i ( 1 - \mathrm{e}^{i(k+\alpha)} )\,.
    \end{eqnarray}

    Eventually, let us also notice that the previous results can be identically extended to the situation in which we consider on each interval the sum of the Laplace-Beltrami operator and a continuous potential bounded from below $v_i^a(x)$ which has to be periodic, i.e. $v^a_{i+1}(x) = v_i^a(x)$. Indeed, such a function does not affect the analysis of self-adjoint extensions which remain those of the Laplace-Beltrami operator.

    Before concluding this section, let us remark that the above discussion only shows that the calculated functions are candidates to be generalised eigenfunctions. However, to actually be a generalised eigenfunction they also need to satisfy an additional property: its inner product with every function in the domain of the Laplacian must be finite. Therefore, even if one can impose the boundary conditions for every $k$ in some interval of the positive line, some of these values might not be in the spectrum of the associated Laplace-Beltrami operator. In the figures \ref{fig:eigenfunctions_1}, \ref{fig:eigenfunctions_2} and \ref{fig:eigenfunctions_3} some of the candidates to generalised eigenfunctions are shown for different combinations of $k$ and $\alpha^i_j$.

    \begin{figure}[bp]
        \centering
        \begin{subfigure}[b]{0.9\textwidth}
        \centering
            \includegraphics[width=1\textwidth]{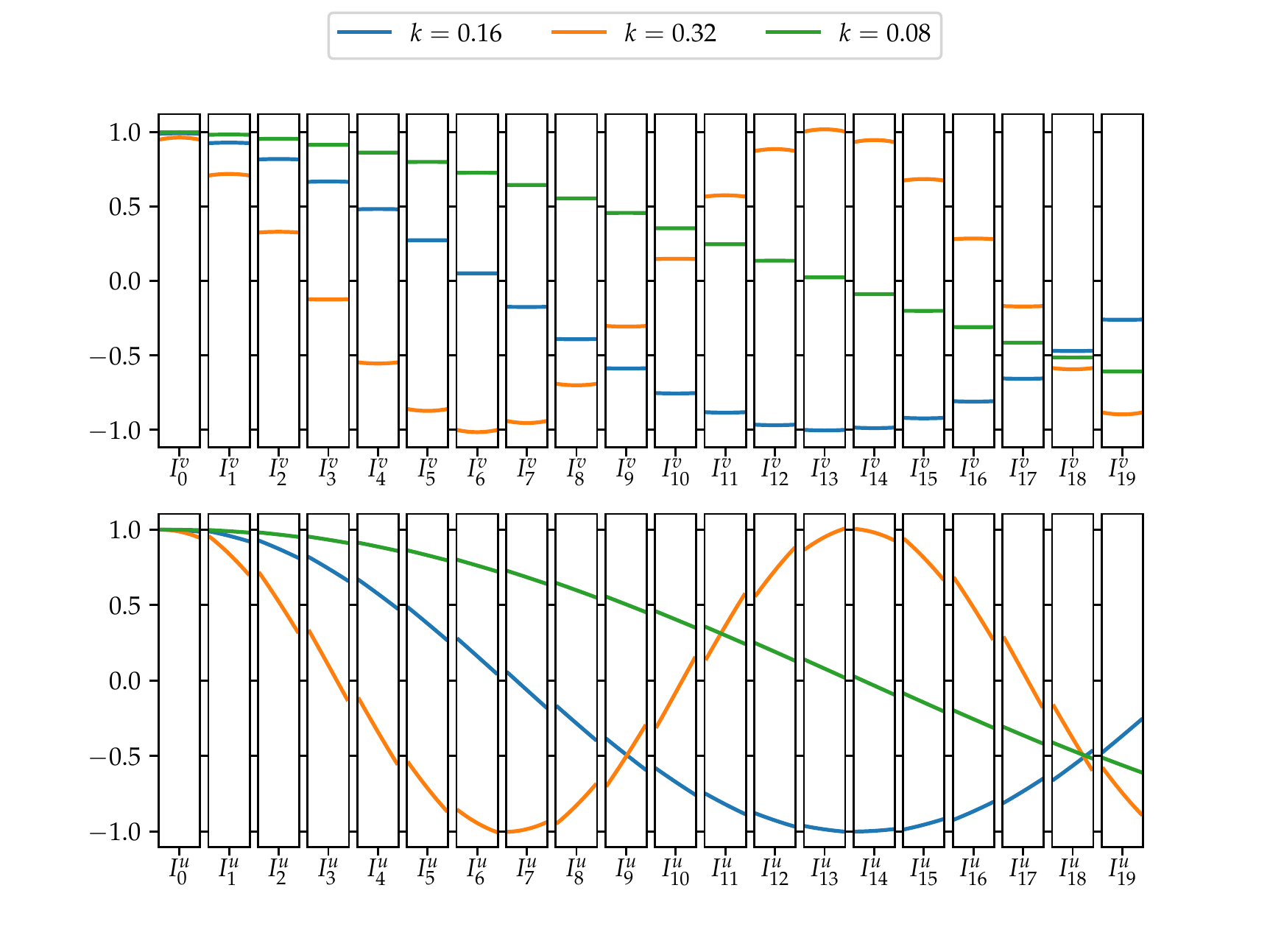}
            \vspace*{-3em}\caption{\small{Real part.}} \label{fig:eigenfunctions_1_a}
        \end{subfigure} \vspace*{2em}

        \begin{subfigure}[b]{0.9\textwidth}
        \centering
            \includegraphics[width=1\textwidth]{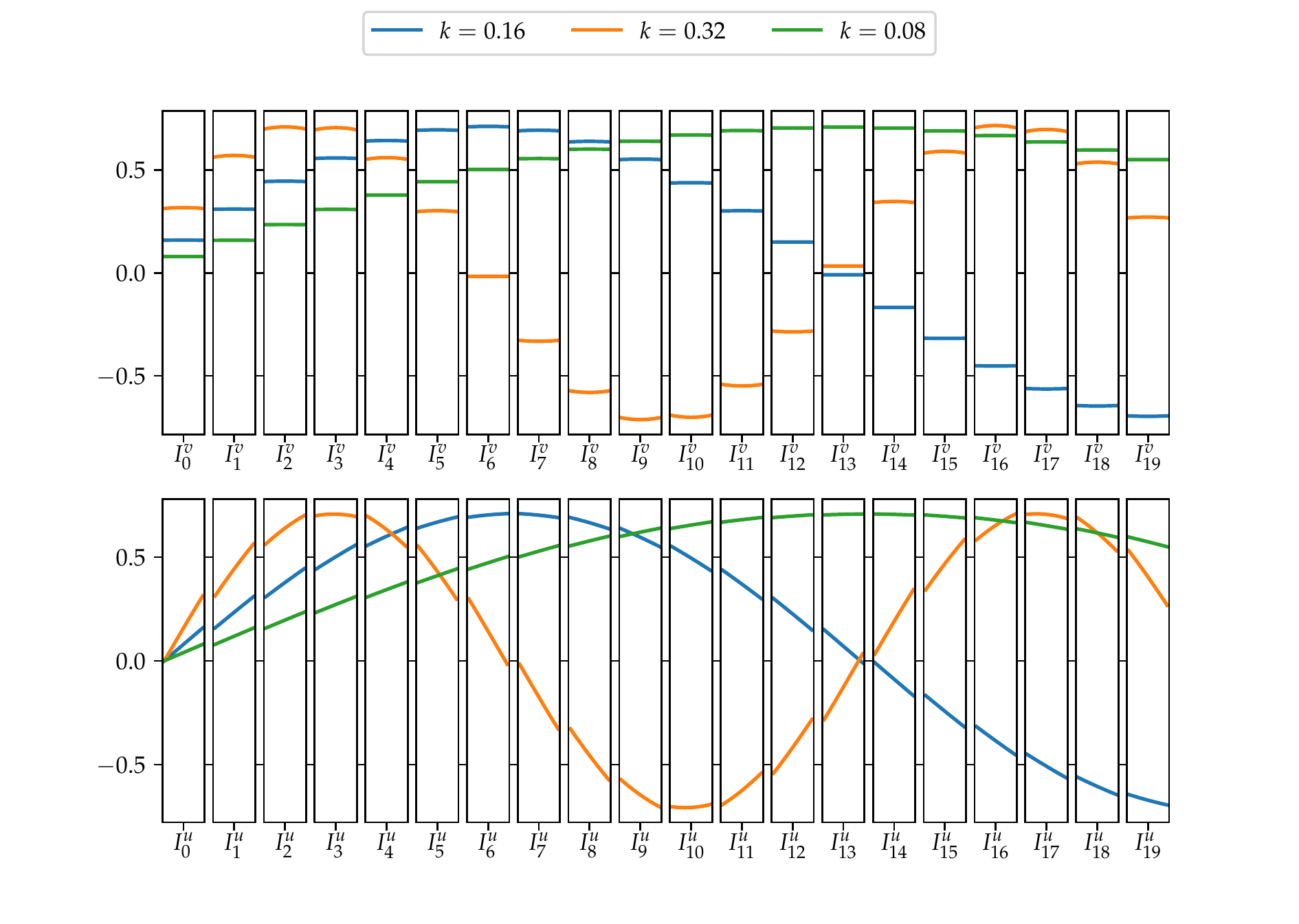}
            \vspace*{-3em}\caption{\small{Imaginary part.}} \label{fig:eigenfunctions_1_b}
        \end{subfigure}\vspace*{1em}
        \caption{\small{Real and imaginary parts of a generalised eigenfunction for $\alpha_j^i = 0$ and several values of $k$. For each of the images, the upper row shows the value on the loops while the lower row shows the value in the chain.}} \label{fig:eigenfunctions_1}
    \end{figure}

    \begin{figure}[bp]
        \centering
        \includegraphics[width=0.9\textwidth]{"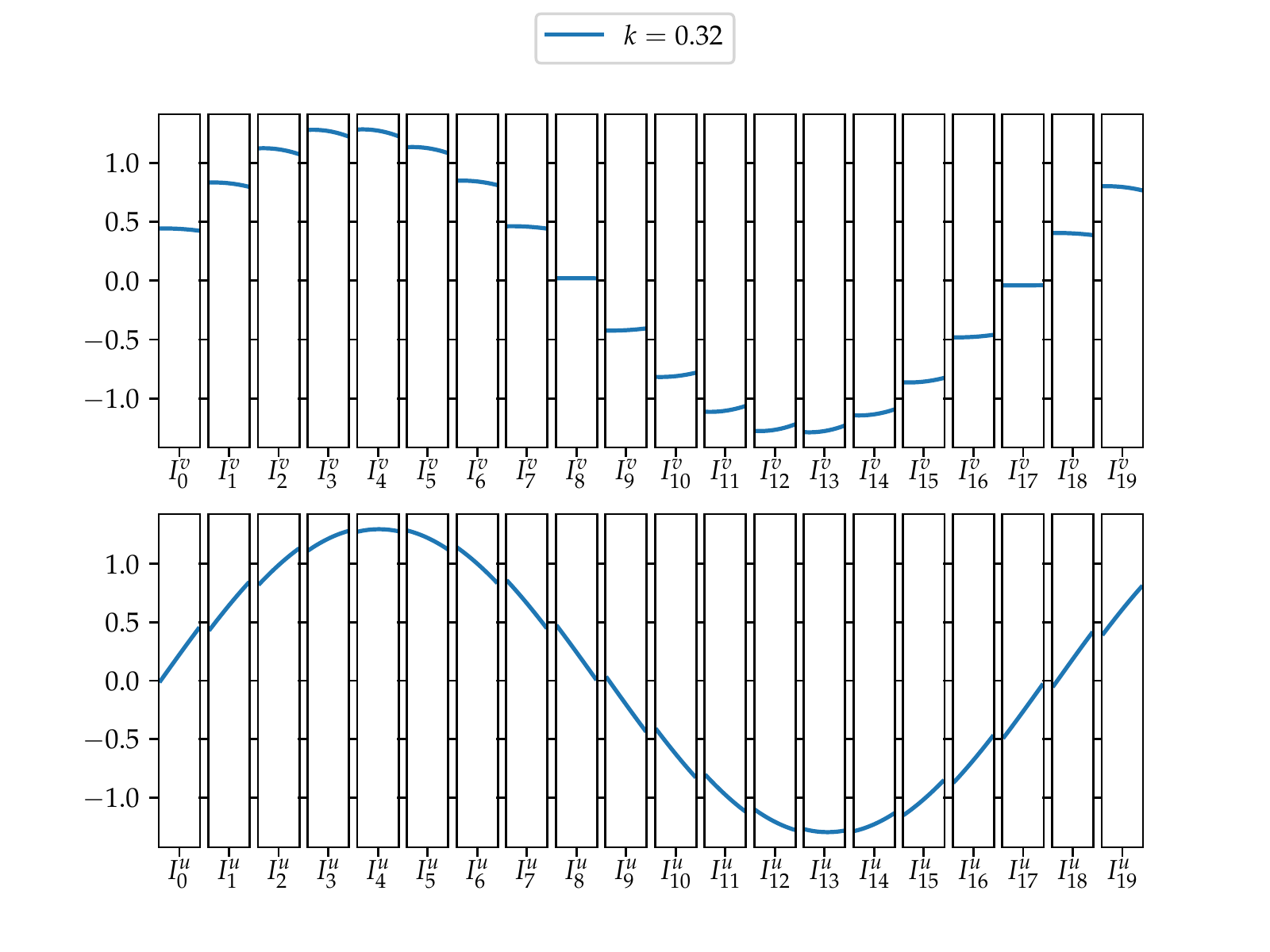}
        \caption{\small{Value of a generalised eigenfunction for $k = 1/\pi$, $\alpha_3^i = \alpha_1^i = 0$, $\alpha_2^i = 0.9 / \pi$. The upper row shows the value on the loops while the lower row shows the value in the chain.}} \label{fig:eigenfunctions_2}
    \end{figure}

    \begin{figure}[bp]
        \centering
        \includegraphics[width=0.9\textwidth]{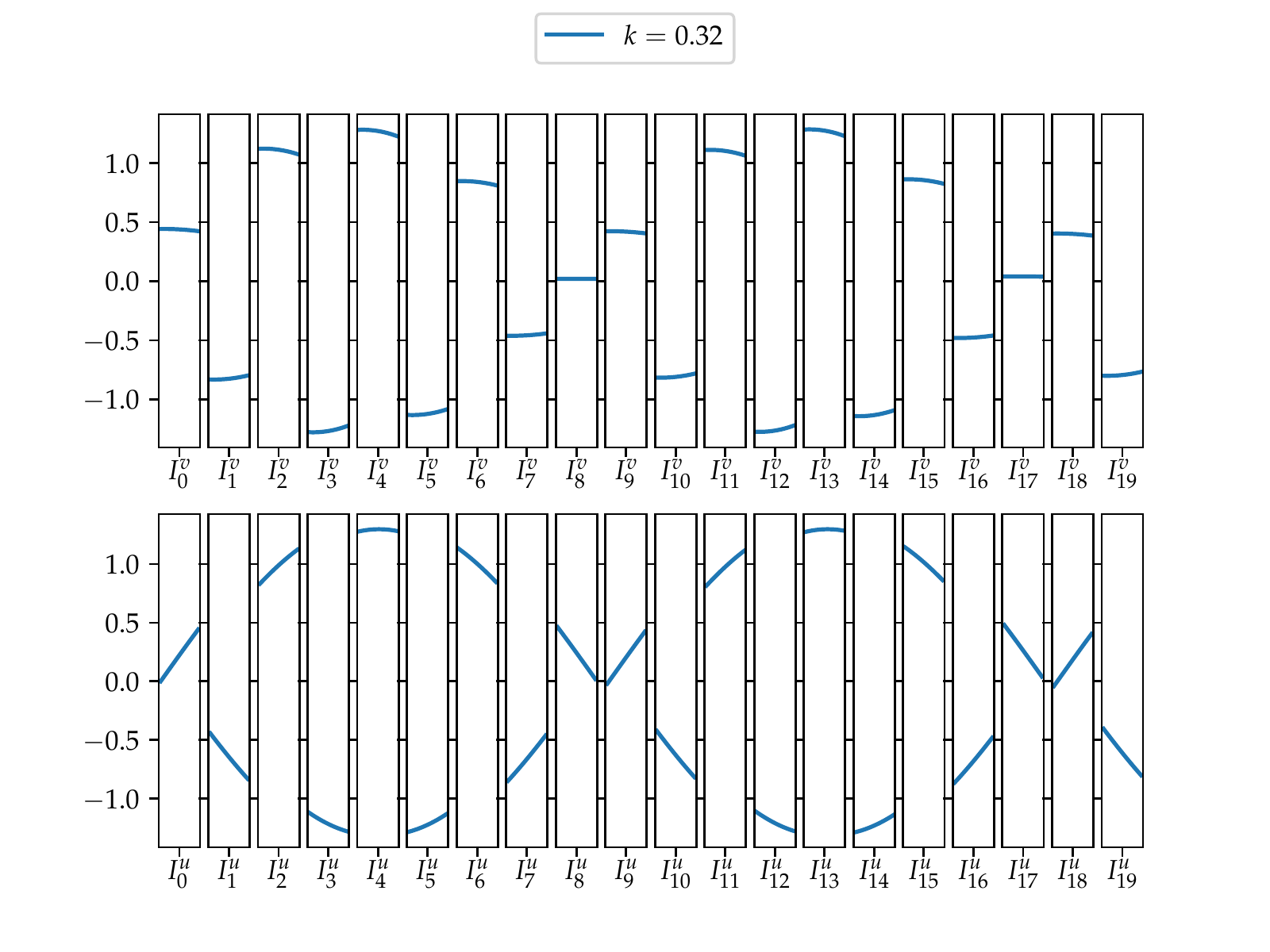}
        \caption{\small{Value of a generalised eigenfunction for $k = 1/\pi$, $\alpha_3^i = \pi$, $\alpha_1^i = 0$, $\alpha_2^i = 0.9 / \pi$. The upper row shows the value on the loops while the lower row shows the value in the chain.}} \label{fig:eigenfunctions_3}
    \end{figure}


\section{Conclusions and discussion}\label{sec:discussion}
    In this paper we have addressed the issue of the existence of self-adjoint extensions of the Laplace operator on quantum circuits compatible with the action of a group of symmetries. In particular, we have investigated a class of circuits obtained as a chain of small unit cells, each  one forming a finite graph. The simple example of unit cells made up of single loops has been fully analysed and different unitary representations of the group of integers $\mathbb{Z}$ on the Hilbert space associated with the quantum system have been provided. Eventually, the features of the $\mathbb{Z}$-invariant self-adjoint extensions of the Laplace-Beltrami operator with respect to these unitary representations are shown and a way to obtain the generalised eigenfunctions is outlined. The generalisation to more complex unit cells follows straightforwardly. Moreover, the method can be applied without modification to another important situation: the case of closed, compact quantum circuits. For these systems the same kind of analysis leads to similar conclusions provided that the group under study respects the symmetry of the circuit itself. For instance, one could have considered a chain with a finite number, say $m$, of repeated cells, with periodic boundary conditions. In this case the cyclic group of order $m$, $\mathbb{Z}_m$, would act in a similar way as the group $\mathbb{Z}$ on an infinite chain. In this case the spectrum of the Laplacian operator would be discrete, as the carrier space of the quantum system is compact.

    The analysis presented in this work can be considered as a preliminary work towards a complete modelisation of quantum circuits. The attention paid to this topic has recently grown because quantum systems on graphs provide a setting for universal computation \cite{Childs2009}. Therefore, coming back to the chain described in this work, any unit cell could be thought of as a gate of a more complex quantum computer, and the family of quasi-$\delta$ self-adjoint extensions provides a set of vertex connections which preserve the underlying topology of the circuit and can be compatible with translational symmetry, too.

    As a final remark let us stress the role played by the group $\mathbb{Z}$ as a symmetry of the quantum system. Indeed, using symmetry considerations we have been able to identify a wide class of self-adjoint extensions which preserves the topology of the graph (these self-adjoint extensions are obtained by imposing boundary conditions already known in the literature as quasi-$\delta$ boundary conditions \cite{BalmasedaPerezPardo2019}). Notice that the family of self-adjoint extensions obtained is more general than a merely periodic repetition of the parameters defining the self-adjoint extension at each unit cell. Due to the group of symmetries, not all the parameters define $\mathbb{Z}$-invariant self-adjoint extensions for some representation $V$ of the group $\mathbb{Z}$. The parameters $\alpha$ and $\delta$ introduced in the last section, indeed, have to be constant. Such an action has also allowed us to compute more easily the generalised eigenfunctions only knowing the eigenfunctions on an interval. It would be interesting to see the difference in the spectral properties when a periodic potential is added, and compare these results with those deriving from Bloch theory of periodic quantum systems.

    \thanks{The authors acknowledge partial support provided by the ``Ministerio de Economía, Industria y Competitividad'' research project MTM2017-84098-P and QUITEMAD project P2018/TCS-4342 funded by ``Comunidad Autónoma de Madrid''. A.B. acknowledges financial support by ``Universidad Carlos III de Madrid'' through Ph.D. program grant PIPF UC3M 01-1819. F.dC. acknowledges financial support by QUITEMAD project P2018/TCS-4342
    }




\end{document}